\documentclass{article}
\usepackage{style}
\usepackage{shortcuts}

\title{Beyond 3SUM-Hardness for Geometric Problems}
\author{Nick Fischer\thanks{Max Planck Institute for Informatics. Email: \href{mailto:nfischer@mpi-inf.mpg.de}{nfischer@mpi-inf.mpg.de}}}
\date{}

\begin{document}

\maketitle

\begin{abstract}
\noindent
Since the early days of fine-grained complexity, many geometric problems have been shown to require quadratic time under the 3SUM Hypothesis. But as the 3SUM Hypothesis was recently refuted, these lower bounds have lost their foundation. In this note we point out that many of them can be recovered under an alternative hardness assumption, the \emph{$k$-AP Hypothesis}. The required changes to the proofs are trivial, so the lower bounds remain just as easy to teach.
\end{abstract}

\section{Introduction}
Even before ``fine-grained complexity'' became a field with a name, Gajentaan and Overmars~\cite{GO95} established some of its very first conditional lower bounds. Their starting point was the \emph{3SUM Hypothesis}, which postulated that no algorithm can decide in time $O(n^{2-\epsilon})$ (for any constant $\epsilon > 0$) whether a given set of $n$ numbers contains a solution to $a + b = c$. From this assumption, they derived quadratic-time lower bounds for a diverse set of geometric problems: detecting three collinear points, finding separating lines, covering, visibility and motion planning problems, and many more. Follow-up work kept adding to this list~\cite{dBdGO97,ACH+98,BvKT98,Eri99,AHI+01,BH01,SEO03,AEK05,EHM06,CEH07,AH08,CH19,BN21,KN22,GK25,GKST25,ABB+26,BFK+26,LCA26}. Part of the appeal of these lower bounds is their simplicity: most of the reductions are elementary and can be explained with a single picture. Unsurprisingly, they have become standard material in courses on fine-grained complexity, where they are often among the very first lower bounds that students encounter.

But recently, fine-grained complexity witnessed a spectacular breakthrough: Alman and Vassilevska Williams~\cite{AV26} gave a truly subquadratic algorithm for 3SUM, thereby refuting the 3SUM Hypothesis. At the same time and based on the same new technique, the APSP Hypothesis was refuted as well.\footnote{The APSP Hypothesis postulates that no algorithm can compute all pairs of distances in a given $n$-vertex directed edge-weighted graph in time $O(n^{3-\epsilon})$ (for any constant $\epsilon > 0$).} As a consequence, a large body of conditional lower bounds based on these two hypotheses has become \emph{vacuous}. The reductions themselves of course remain correct, but a reduction from 3SUM or APSP no longer constitutes evidence that a problem is hard. For algorithm design, this is good news: the complexity of the many 3SUM- or APSP-hard problems is now open, and for each of them we may hope for a polynomially faster algorithm. But at the same time, many of these problems will likely remain hard for other reasons. Sorting out which is which, either by designing faster algorithms or by restoring hardness results under other credible hypotheses, will be an important task for the fine-grained complexity community now.

In this note we point out that, for many geometric problems, this task is surprisingly easy: many of the classical 3SUM-based lower bounds remain valid, with essentially unchanged proofs, under a different hardness assumption from fine-grained complexity, the \emph{$k$-AP Hypothesis}~\cite{DFG+26}. The observation is simple, and we make no claim of technical depth. The reason for writing this note is more practical: several colleagues I talked to are teaching courses on fine-grained complexity this semester, and many of the lower bounds they had planned to cover have disappeared overnight. This note offers a way to keep many lower bounds for geometric problems on the syllabus, with only small changes. We hope that it will be useful to anyone teaching fine-grained complexity, now or in the future.

\paragraph{Arithmetic Progression Detection}
A \emph{$k$-term arithmetic progression} (or \emph{$k$-AP} for short) is a sequence of numbers of the form $a, a + d, a + 2d, \dots, a + (k-1) d$; we call it \emph{nontrivial} if $d \neq 0$. In the \emph{$k$-AP Detection} problem the task is to decide whether a given set $A \subseteq \set{1, \dots, N}$ contains a nontrivial $k$-AP. (Note that~$N$ bounds the universe, and not just the size of $A$.) The trivial algorithm enumerates the first two terms~$a, a + d \in A$ and tests whether the remaining $k - 2$ terms are present in $A$, in total time $O(N^2)$. For $k = 3$ there is a much faster $O(N \log N)$-time algorithm based on the Fast Fourier Transform. For $k \geq 4$, however, no polynomial improvement over the trivial algorithm is known, which motivated the following hypothesis for all constants $k \geq 4$~\cite{DFG+26}.

\begin{hypothesis}[$k$-AP Hypothesis]
There is no algorithm that decides whether a given set $A \subseteq \set{1, \dots, N}$ contains a nontrivial $k$-AP in time $O(N^{2-\epsilon})$, for any $\epsilon > 0$.
\end{hypothesis}

As argued before, the smallest interesting case is $k = 4$. Moreover, it is known that the $k$-AP Hypothesis implies the $(k+1)$-AP Hypothesis~\cite{DFG+26}, and so the $4$-AP Hypothesis is the strongest member of this family. The lower bounds in this note hold even under the weakest one: it suffices that the $k$-AP Hypothesis holds for \emph{some} constant $k \geq 4$.

Why should we believe in the $k$-AP Hypothesis, especially now that 3SUM has been refuted? First, let us take a step back. Fine-grained complexity has never been about hardness for its own sake. Its ultimate goal is to find the fastest possible algorithm for every problem, and its hypotheses serve as signposts that tell us where \emph{not} to look: as long as a problem is hard under some hypothesis, any faster algorithm for it would refute the hypothesis, so there is little point in attacking the problem directly while the hypothesis stands. In this view, whether a hypothesis is ultimately true is less important than whether it faithfully marks the limits of our current algorithmic techniques. Read in this way, this note puts many of the signposts for geometric problems back in place: whoever hopes to solve one of these problems faster should first find a faster algorithm for detecting arithmetic progressions. Second, there is good reason to trust this particular signpost: despite the problem's natural appeal and its close connection to additive combinatorics, the fastest known algorithms for $k$-AP Detection improve upon the trivial running time by merely subpolynomial factors~\cite[Section~5]{DFG+26}.\footnote{The fastest known algorithm has in fact improved since~\cite{DFG+26}. OpenAI recently announced quasipolynomial bounds for Szemer\'edi's theorem~\cite{OpenAI26}: for every $k \geq 3$, every set $A \subseteq \set{1, \dots, N}$ without a nontrivial $k$-AP has size $|A| \leq N / 2^{\Omega((\log N)^c)}$, for some constant $c = c(k) > 0$. This immediately yields a simple algorithm: if $A$ is larger than this bound, the answer is yes; otherwise, the brute-force algorithm runs in time $O(N + |A|^2) = N^2 / 2^{\Omega((\log N)^c)}$.} Third, the recent breakthrough~\cite{AV26} does not seem to extend to $4$-APs easily. At its core, their algorithm is still a Strassen-like algebraic algorithm. Algorithms of this type are tailored to \emph{trilinear} problems, such as matrix multiplication or $3$-AP Detection. In contrast, $4$-AP Detection is inherently a \emph{$4$-linear} problem, and to the best of our knowledge Strassen-like techniques have so far \emph{never} been extended to $4$-linear problems in a nontrivial way. See~\cite[Section~5]{DFG+26} for a more detailed discussion of the $k$-AP Hypothesis.

\paragraph{Geometric Problems}
Let us finally return to the geometric problems. Gajentaan and Overmars~\cite{GO95} routed most of their reductions through an intermediate problem called \emph{GeomBase}: Given $n$ points with integer coordinates on the three horizontal lines $y = 0$, $y = 1$ and $y = 2$, is there a non-horizontal line that contains three of the points? GeomBase is merely 3SUM in disguise (the points $(a, 0)$, $(c, 1)$ and $(b, 2)$ are collinear if and only if $a + b = 2c$~\cite[Theorem~3.2]{GO95}), but its geometric formulation makes it a convenient starting point for further reductions. In the same spirit, we consider the following natural generalization from three to $k$ lines.

\begin{definition}[$k$-GeomBase]
Given a set of $n$ points with integer coordinates on the $k$ horizontal lines $y = 1, \dots, k$, determine whether there exists a non-horizontal line containing $k$ of the points.
\end{definition}

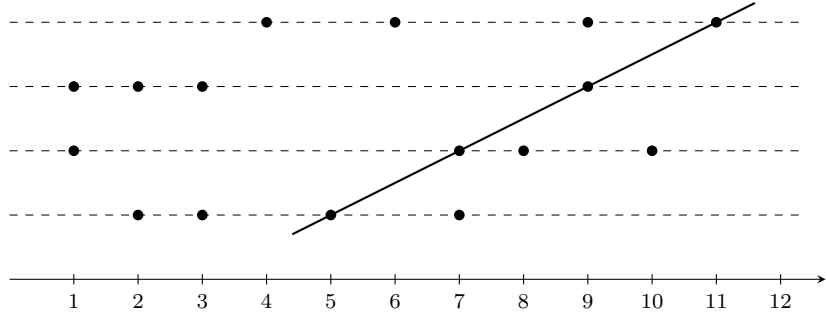
\begin{figure}[t]
\centering
\medskip
\begin{tikzpicture}[x=0.85cm, y=0.85cm, >=stealth, black]
    \draw[->] (0, 0) -- (12.7, 0);
    \foreach \x in {1, ..., 12}
        \draw (\x, 0.08) -- (\x, -0.08) node[below] {\footnotesize $\x$};
    \foreach \y in {1, ..., 4}
        \draw[very thin, dashed] (0, \y) -- (12.4, \y);
    \draw[thick] (4.4, 0.7) -- (11.6, 4.3);
    \foreach \x/\y in {2/1, 3/1, 5/1, 7/1, 1/2, 7/2, 8/2, 10/2, 1/3, 2/3, 3/3, 9/3, 4/4, 6/4, 9/4, 11/4}
        \fill (\x, \y) circle (2pt);
\end{tikzpicture}
\medskip
\caption{A yes-instance of $4$-GeomBase. The solid line contains the four points $(5, 1), (7, 2), (9, 3), (11, 4)$, one on each line; correspondingly, $5, 7, 9, 11$ form a $4$-AP.}
\label{fig:geombase}
\end{figure}

See \cref{fig:geombase} for an example. Just as GeomBase is a reformulation of 3SUM, $k$-GeomBase is a reformulation of $k$-AP Detection: $k$ points $(a_1, 1), \dots, (a_k, k)$ lie on a common line if and only if $a_1, \dots, a_k$ form a $k$-AP. In particular, the following lemma is immediate.

\begin{lemma} \label{lem:k-geombase}
For any constant $k \geq 4$, the $k$-GeomBase problem cannot be solved in time $O(n^{2-\epsilon})$ (for any constant $\epsilon > 0$), unless the $k$-AP Hypothesis fails.
\end{lemma}
\begin{proof}
Let $A$ be the given $k$-AP instance. We randomly partition $A$ into sets $A_1, \dots, A_k$ (we apply ``color coding''). Then construct the $k$-GeomBase instance
\begin{equation*}
    P = \set{ (a, i) : 1 \leq i \leq k,\, a \in A_i }.
\end{equation*}
Whenever there are $k$ points $(a_1, 1), \dots, (a_k, k) \in P$ that lie on a shared line, the integers $a_1, \dots, a_k$ are equispaced and thus form a $k$-AP in $A$. Since $a_1, \dots, a_k$ are pairwise distinct (by color coding), the $k$-AP is nontrivial. Conversely, suppose that there is a nontrivial $k$-AP $a_1, \dots, a_k \in A$. With constant probability we color it correctly, $(a_1, \dots, a_k) \in A_1 \times \dots \times A_k$, and in that event $P$ has the solution $(a_1, 1), \dots, (a_k, k)$. The color coding technique can be derandomized using standard tools~\cite{AYZ95}.
\end{proof}

The simple observation of this note is that all reductions via GeomBase that we are aware of work just as well when starting from $k$-GeomBase. Intuitively, the number of lines never plays any role in these reductions. The typical construction turns each of the horizontal lines into a \emph{wall} of horizontal segments that leaves small \emph{holes} precisely at the given points~\cite[Theorem~5.1]{GO95}. The object in question (a separating line, a line of sight, a rod moving from one side to the other, and so on) must then pass through a hole in every wall, and this is possible if and only if some non-horizontal line hits a point on each of the original lines. Whether there are three walls or four or $k$ is irrelevant for this argument. Collecting the problems that were shown to be 3SUM-hard via GeomBase in~\cite{GO95,dBdGO97,BvKT98,AEK05,BFK+26}, we thus obtain the following \cref{thm:main}.

\begin{theorem} \label{thm:main}
None of the following problems can be solved in time $O(n^{2-\epsilon})$ (for any constant $\epsilon > 0$), unless the $k$-AP Hypothesis fails for all constants $k \geq 4$.
\begin{itemize}
    \item {\normalfont\bfseries Separator of Axis-Parallel Segments:}
    Given $n$ disjoint axis-parallel segments in the plane, decide whether some line avoids all segments and has segments on both sides~\cite[Theorem~5.2]{GO95}.

    \item {\normalfont\bfseries Strips Cover a Box:} Given $n$ strips and an axis-parallel rectangle in the plane, decide whether the strips cover the rectangle~\cite[Theorem~6.1]{GO95}.

    \item {\normalfont\bfseries Strips Cover a Triangle:}
    Given $n$ strips and a triangle in the plane, decide whether the strips cover the triangle~\cite[Theorem~6.2]{GO95}.

    \item {\normalfont\bfseries Arbitrarily Oriented Rectangles Cover a Rectangle:}
    Given $n$ arbitrarily oriented rectangles and an axis-parallel rectangle $R$ in the plane, decide whether the $n$ rectangles cover~$R$~\cite[Section~6.1]{GO95}.

    \item {\normalfont\bfseries Triangles Cover a Triangle:}
    Given $n$ triangles and a designated triangle $t$ in the plane, decide whether the $n$ triangles cover $t$~\cite[Theorem~6.2]{GO95}.

    \item {\normalfont\bfseries Hole in a Union of Triangles:}
    Given $n$ triangles in the plane, decide whether their union has a hole \cite[Theorem~6.3]{GO95}.

    \item {\normalfont\bfseries Area of a Union of Triangles:}
    Given $n$ triangles in the plane, compute the area of their union \cite[Theorem~6.5]{GO95}.

    \item {\normalfont\bfseries Maximum Depth of Halfplanes:}
    Given $n$ halfplanes and an integer $t$, decide whether some point lies in at least $t$ of them~\cite[Theorem~6.6]{GO95}.

    \item {\normalfont\bfseries Visibility Between Two Segments:}
    Given $n$ horizontal segments and two further horizontal segments $s, s'$ in the plane, decide whether some point on $s$ sees some point on $s'$~\cite[Theorem~7.1]{GO95}.

    \item {\normalfont\bfseries Visibility of a Segment from Infinity:}
    Given $n$ axis-parallel segments and a horizontal segment $s$ in the plane, decide whether some ray starting on $s$ avoids all other segments~\cite[Theorem~7.2]{GO95}.

    \item {\normalfont\bfseries Visibility of a Triangle from a Point in 3D:}
    Given $n$ horizontal triangles, a horizontal triangle $t$ and a point $p$ in $\Real^3$, decide whether $p$ sees some point of $t$~\cite[Theorem~7.3]{GO95}.

    \item {\normalfont\bfseries Planar Motion Planning for a Rod:}
    Given $n$ disjoint axis-parallel segments in the plane, decide whether a rod can be moved (by translations and rotations) between two given placements while avoiding the segments~\cite[Theorem~8.1]{GO95}.

    \item {\normalfont\bfseries Translational Motion Planning for a Vertical Rod in 3D:}
    Given $n$ disjoint horizontal triangles in $\Real^3$, decide whether a vertical rod can be translated between two given positions while avoiding the triangles~\cite[Theorem~8.2]{GO95}.

    \item {\normalfont\bfseries Perfect BSP for Disjoint Planar Segments:}
    Given $n$ pairwise disjoint line segments in the plane, decide whether a binary space partition can recursively split cells by lines, without cutting any input segment, until each leaf cell contains at most one segment~\cite{dBdGO97}.

    \item {\normalfont\bfseries One-Fillability of a Polyhedral Mold:}
    Given a polyhedron of complexity $n$, decide whether there is an orientation in which it can be filled through a single opening without trapping air~\cite{BvKT98}.

    \item {\normalfont\bfseries Narrowest Color-Spanning Strip:} Given $n$ colored points in the plane, compute the narrowest strip that contains a point of every color~\cite{AHI+01}.

    \item {\normalfont\bfseries Cardinality of the Distant-Horizon Edge Set:}
    Given a piecewise-linear terrain with $n$ edges, count the edges that appear on its distant horizon (its upper silhouette under horizontal parallel projection) for some viewing direction~\cite{AEK05}.

    \item {\normalfont\bfseries Stabbing and Avoiding Segments:}
    Given $n$ segments in the plane, partitioned into sets $S$ and $A$, decide whether some line intersects all segments in $S$ and avoids all segments in $A$~\cite{BFK+26}.

    \item {\normalfont\bfseries Double-Wedge Intersection:}
    Given $n$ double-wedges in the plane, decide whether they have a common point~\cite{BFK+26}.
\end{itemize}
\end{theorem}

We omit the proofs, as each item follows trivially by inspecting the reduction in the respective reference and starting from $k$-GeomBase instead of GeomBase. In \cref{sec:motion-planning} we work out one example in detail.

\paragraph{Beyond GeomBase}
Not all 3SUM-based hardness proofs for geometric problems go through GeomBase, e.g.~\cite{GO95,ACH+98,Eri99,BH01,SEO03,EHM06,CEH07,AH08,CH19,BN21,KN22,GK25,GKST25,ABB+26}, and those that do not are not covered by this note. An important example with a direct reduction from 3SUM is the problem of detecting three collinear points among $n$ points in the plane~\cite[Theorem~4.1]{GO95}.

\section{Example: Planar Motion Planning} \label{sec:motion-planning}
To illustrate how the lower bounds in \cref{thm:main} are obtained, let us work out one example in detail: \emph{planar motion planning}. Imagine a robot that is supposed to navigate through an environment of obstacles. In a classical geometric model, the robot is a convex polygon of constant complexity in the plane, and the $n$ obstacles are line segments. Given a start placement and a goal placement of the robot, the task is to decide whether the robot can move from the start to the goal, by translations and rotations, without ever colliding with an obstacle.

This classical problem can be solved in near-quadratic time~$\tilde O(n^2)$~\cite{Veg90,CK93,HS96}. Gajentaan and Overmars~\cite{GO95} showed that any polynomial improvement would refute the 3SUM Hypothesis, even if the robot is a rod and the obstacles are disjoint axis-parallel segments. Their reduction works equally well when starting from $k$-GeomBase, which yields the following special case of \cref{thm:main}.

\begin{theorem} \label{thm:motion-planning}
Planar motion planning for a rod moving through $n$ disjoint axis-parallel segments cannot be solved in time $O(n^{2-\epsilon})$ (for any constant $\epsilon > 0$), unless the $k$-AP Hypothesis fails for all constants $k \geq 4$.
\end{theorem}

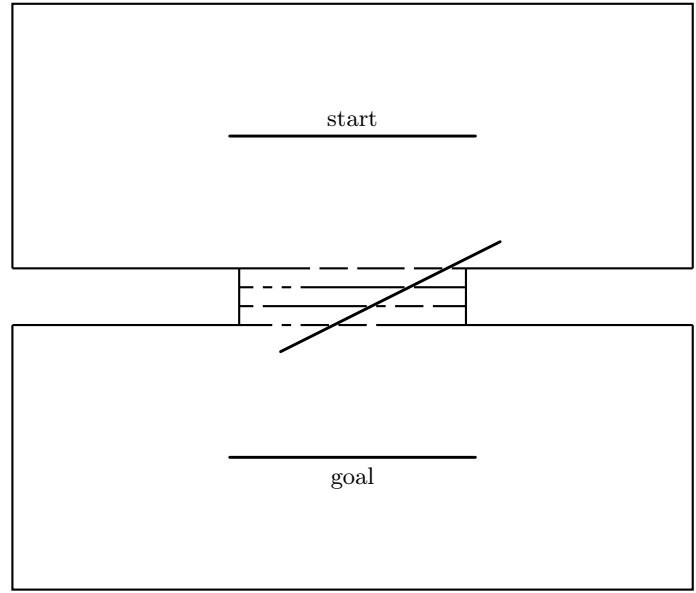
\begin{figure}[htb]
\centering
\medskip
\begin{tikzpicture}[x=0.25cm, y=0.25cm, black]
    % rooms and chamber
    \draw[thick] (-12, 4) -- (-12, 18) -- (24, 18) -- (24, 4);
    \draw[thick] (-12, 1) -- (-12, -13) -- (24, -13) -- (24, 1);
    \draw[thick] (0, 1) -- (0, 4);
    \draw[thick] (12, 1) -- (12, 4);
    % walls with holes
    \foreach \l/\r in {-12/1.75, 2.25/2.75, 3.25/4.75, 5.25/6.75, 7.25/24}
        \draw[thick] (\l, 1) -- (\r, 1);
    \foreach \l/\r in {0/0.75, 1.25/6.75, 7.25/7.75, 8.25/9.75, 10.25/12}
        \draw[thick] (\l, 2) -- (\r, 2);
    \foreach \l/\r in {0/0.75, 1.25/1.75, 2.25/2.75, 3.25/8.75, 9.25/12}
        \draw[thick] (\l, 3) -- (\r, 3);
    \foreach \l/\r in {-12/3.75, 4.25/5.75, 6.25/8.75, 9.25/10.75, 11.25/24}
        \draw[thick] (\l, 4) -- (\r, 4);
    % placements of the rod
    \draw[line width=1.1pt, line cap=round] (-0.5, 11) -- node[above] {\small start} (12.5, 11);
    \draw[line width=1.1pt, line cap=round] (2.19, -0.41) -- (13.81, 5.41);
    \draw[line width=1.1pt, line cap=round] (-0.5, -6) -- node[below] {\small goal} (12.5, -6);
\end{tikzpicture}
\medskip
\caption{The reduction from $4$-GeomBase to planar motion planning, applied to the instance from \cref{fig:geombase}. The rod has to move from its start placement in the top room to its goal placement in the bottom room. Since the four holes at $(5, 1), (7, 2), (9, 3), (11, 4)$ are collinear, the rod can slide through them, as indicated by the intermediate placement.}
\label{fig:motion-planning}
\end{figure}

\begin{proof}
Fix a constant $k \geq 4$, and let $P \subseteq \Int \times \set{1, \dots, k}$ be a $k$-GeomBase instance consisting of $n$ points. By shifting, we may assume that all $x$-coordinates lie in $\set{1, \dots, X}$ for some integer $X$. We construct the following environment of obstacles; see \cref{fig:motion-planning}.

The core of the construction consists of $k$ horizontal \emph{walls}, where the $i$-th wall lies at height $y = i$. Into each wall we cut a small \emph{hole} $(a - \frac14, a + \frac14) \times \set{i}$ for every point $(a, i) \in P$. The inner walls $2, \dots, k - 1$ span a narrow \emph{chamber} $[0, X + 1] \times [1, k]$ whose left and right sides are closed by vertical segments. The first and the $k$-th wall extend beyond the chamber and form the ceiling of a large \emph{bottom room} and the floor of a large \emph{top room}, respectively. The robot is a rod whose length $L$ exceeds the diagonal of the chamber, and both rooms have width and height at least $3L$, say, so that the rod can move and rotate freely within each room. The task is to move the rod from a start placement in the middle of the top room to a goal placement in the middle of the bottom room.\footnote{Strictly speaking, the segments in this construction are not disjoint, as they touch at the corners of the rooms and the chamber. To fix this, we can slightly shorten the segments there and add a constant number of further segments, so that each resulting gap has a sharp bend that the long rod cannot squeeze through. We ignore this technicality here.} In total, the environment consists of $O(n)$ segments and can be constructed in time $O(n \log n)$.

If some non-horizontal line $\ell$ hits $k$ points of $P$, then $\ell$ passes through a hole in every wall. Hence, we can move the rod onto $\ell$ within the top room, slide it along $\ell$ through all $k$ holes into the bottom room, and finally move it to its goal placement. Conversely, since the rod is longer than the diagonal of the chamber, at every moment of a collision-free motion it reaches into the top room or into the bottom room. By continuity, at some moment it reaches into both rooms, and thus passes through one hole in every wall. As the holes are narrow and centered at integer coordinates, this is only possible if the corresponding $k$ points of $P$ are collinear (as in~\cite[Theorem~5.1]{GO95}). This proves that the reduction is correct.

All in all, an $O(n^{2-\epsilon})$-time algorithm for planar motion planning would solve $k$-GeomBase in time $O(n^{2-\epsilon})$ and thereby refute the $k$-AP Hypothesis by \cref{lem:k-geombase}.
\end{proof}

\section*{AI Disclosure}
I used Claude Opus 5.5 and ChatGPT Astra 6 to assist with the write-up, the generation of figures, and the search for references, but not for any of the mathematical ideas.

\bibliographystyle{alphaurl}
\bibliography{main}

@article{GO95,
  author       = {Anka Gajentaan and
                   Mark H. Overmars},
  title        = {{On a Class of $O(n^2)$ Problems in Computational Geometry}},
  journal      = {Comput. Geom.},
  volume       = {5},
  pages        = {165--185},
  year         = {1995},
  doi          = {10.1016/0925-7721(95)00022-2},
  url          = {https://doi.org/10.1016/0925-7721(95)00022-2},
  bibsource    = {dblp computer science bibliography, https://dblp.org}
}

@article{dBdGO97,
  author       = {Mark de Berg and
                   Marko de Groot and
                   Mark H. Overmars},
  title        = {{Perfect Binary Space Partitions}},
  journal      = {Comput. Geom.},
  volume       = {7},
  pages        = {81--91},
  year         = {1997},
  doi          = {10.1016/0925-7721(95)00045-3},
  url          = {https://doi.org/10.1016/0925-7721(95)00045-3},
  bibsource    = {dblp computer science bibliography, https://dblp.org}
}

@article{BvKT98,
  author       = {Prosenjit Bose and
                   Marc J. van Kreveld and
                   Godfried T. Toussaint},
  title        = {{Filling polyhedral molds}},
  journal      = {Comput. Aided Des.},
  volume       = {30},
  number       = {4},
  pages        = {245--254},
  year         = {1998},
  doi          = {10.1016/S0010-4485(97)00075-4},
  url          = {https://doi.org/10.1016/S0010-4485(97)00075-4},
  bibsource    = {dblp computer science bibliography, https://dblp.org}
}

@article{AEK05,
  author       = {Daniel Archambault and
                   William S. Evans and
                   David G. Kirkpatrick},
  title        = {{Computing the Set of all the Distant Horizons of a Terrain}},
  journal      = {Int. J. Comput. Geom. Appl.},
  volume       = {15},
  number       = {6},
  pages        = {547--563},
  year         = {2005},
  doi          = {10.1142/S0218195905001841},
  url          = {https://doi.org/10.1142/S0218195905001841},
  bibsource    = {dblp computer science bibliography, https://dblp.org}
}

@article{ACH+98,
  author       = {Esther M. Arkin and
                   Yi-Jen Chiang and
                   Martin Held and
                   Joseph S. B. Mitchell and
                   Vera Sacrist{\'a}n and
                   Steven Skiena and
                   Tae-Heng Yang},
  title        = {{On Minimum-Area Hulls}},
  journal      = {Algorithmica},
  volume       = {21},
  number       = {1},
  pages        = {119--136},
  year         = {1998},
  doi          = {10.1007/PL00009204},
  url          = {https://doi.org/10.1007/PL00009204},
  bibsource    = {dblp computer science bibliography, https://dblp.org}
}

@article{Eri99,
  author       = {Jeff Erickson},
  title        = {{New Lower Bounds for Convex Hull Problems in Odd Dimensions}},
  journal      = {SIAM J. Comput.},
  volume       = {28},
  number       = {4},
  pages        = {1198--1214},
  year         = {1999},
  doi          = {10.1137/S0097539797315410},
  url          = {https://doi.org/10.1137/S0097539797315410},
  bibsource    = {dblp computer science bibliography, https://dblp.org}
}

@article{BH01,
  author       = {Gill Barequet and
                   Sariel Har-Peled},
  title        = {{Polygon Containment and Translational Min-Hausdorff-Distance Between Segment Sets are 3SUM-Hard}},
  journal      = {Int. J. Comput. Geom. Appl.},
  volume       = {11},
  number       = {4},
  pages        = {465--474},
  year         = {2001},
  doi          = {10.1142/S0218195901000596},
  url          = {https://doi.org/10.1142/S0218195901000596},
  bibsource    = {dblp computer science bibliography, https://dblp.org}
}

@article{SEO03,
  author       = {Michael A. Soss and
                   Jeff Erickson and
                   Mark H. Overmars},
  title        = {{Preprocessing chains for fast dihedral rotations is hard or even impossible}},
  journal      = {Comput. Geom.},
  volume       = {26},
  number       = {3},
  pages        = {235--246},
  year         = {2003},
  doi          = {10.1016/S0925-7721(02)00156-6},
  url          = {https://doi.org/10.1016/S0925-7721(02)00156-6},
  bibsource    = {dblp computer science bibliography, https://dblp.org}
}

@article{EHM06,
  author       = {Jeff Erickson and
                   Sariel Har-Peled and
                   David M. Mount},
  title        = {{On the Least Median Square Problem}},
  journal      = {Discret. Comput. Geom.},
  volume       = {36},
  number       = {4},
  pages        = {593--607},
  year         = {2006},
  doi          = {10.1007/S00454-006-1267-6},
  url          = {https://doi.org/10.1007/S00454-006-1267-6},
  bibsource    = {dblp computer science bibliography, https://dblp.org}
}

@article{CEH07,
  author       = {Otfried Cheong and
                   Alon Efrat and
                   Sariel Har-Peled},
  title        = {{Finding a Guard that Sees Most and a Shop that Sells Most}},
  journal      = {Discret. Comput. Geom.},
  volume       = {37},
  number       = {4},
  pages        = {545--563},
  year         = {2007},
  doi          = {10.1007/S00454-007-1328-5},
  url          = {https://doi.org/10.1007/S00454-007-1328-5},
  bibsource    = {dblp computer science bibliography, https://dblp.org}
}

@article{AH08,
  author       = {Boris Aronov and
                   Sariel Har-Peled},
  title        = {{On Approximating the Depth and Related Problems}},
  journal      = {SIAM J. Comput.},
  volume       = {38},
  number       = {3},
  pages        = {899--921},
  year         = {2008},
  doi          = {10.1137/060669474},
  url          = {https://doi.org/10.1137/060669474},
  bibsource    = {dblp computer science bibliography, https://dblp.org}
}

@inproceedings{CH19,
  author       = {Timothy M. Chan and
                   Sariel Har-Peled},
  title        = {{Smallest $k$-Enclosing Rectangle Revisited}},
  booktitle    = {{35th International Symposium on Computational Geometry (SoCG 2019)}},
  series       = {LIPIcs},
  volume       = {129},
  publisher    = {Schloss Dagstuhl - Leibniz-Zentrum f{\"u}r Informatik},
  pages        = {23:1--23:15},
  year         = {2019},
  doi          = {10.4230/LIPICS.SOCG.2019.23},
  url          = {https://doi.org/10.4230/LIPICS.SOCG.2019.23},
  bibsource    = {dblp computer science bibliography, https://dblp.org}
}

@inproceedings{BN21,
  author       = {Karl Bringmann and
                   Andr{\'e} Nusser},
  title        = {{Translating Hausdorff Is Hard: Fine-Grained Lower Bounds for Hausdorff Distance Under Translation}},
  booktitle    = {{37th International Symposium on Computational Geometry (SoCG 2021)}},
  series       = {LIPIcs},
  volume       = {189},
  publisher    = {Schloss Dagstuhl - Leibniz-Zentrum f{\"u}r Informatik},
  pages        = {18:1--18:17},
  year         = {2021},
  doi          = {10.4230/LIPICS.SOCG.2021.18},
  url          = {https://doi.org/10.4230/LIPICS.SOCG.2021.18},
  bibsource    = {dblp computer science bibliography, https://dblp.org}
}

@inproceedings{AHI+01,
  author       = {Manuel Abellanas and
                   Ferran Hurtado and
                   Christian Icking and
                   Rolf Klein and
                   Elmar Langetepe and
                   Lihong Ma and
                   Bel{\'e}n Palop and
                   Vera Sacrist{\'a}n},
  title        = {{Smallest Color-Spanning Objects}},
  booktitle    = {{9th European Symposium on Algorithms (ESA 2001)}},
  series       = {Lecture Notes in Computer Science},
  volume       = {2161},
  publisher    = {Springer},
  pages        = {278--289},
  year         = {2001},
  doi          = {10.1007/3-540-44676-1_23},
  url          = {https://doi.org/10.1007/3-540-44676-1_23},
  bibsource    = {dblp computer science bibliography, https://dblp.org}
}

@article{AYZ95,
  author       = {Noga Alon and
                   Raphael Yuster and
                   Uri Zwick},
  title        = {{Color-Coding}},
  journal      = {J. ACM},
  volume       = {42},
  number       = {4},
  pages        = {844--856},
  year         = {1995},
  doi          = {10.1145/210332.210337},
  url          = {https://doi.org/10.1145/210332.210337},
  bibsource    = {dblp computer science bibliography, https://dblp.org}
}

@inproceedings{DFG+26,
  author       = {Bartlomiej Dudek and
                   Nick Fischer and
                   Geri Gokaj and
                   Ce Jin and
                   Marvin K{\"{u}}nnemann and
                   Xiao Mao and
                   Mirza Redzic},
  title        = {{Classifying Identities: Subcubic Distributivity Checking and Hardness from Arithmetic Progression Detection}},
  booktitle    = {{58th ACM Symposium on Theory of Computing (STOC 2026)}},
  publisher    = {{ACM}},
  pages        = {1453--1464},
  year         = {2026},
  doi          = {10.1145/3798129.3800854},
  url          = {https://doi.org/10.1145/3798129.3800854},
  bibsource    = {dblp computer science bibliography, https://dblp.org}
}

@article{AV26,
  author       = {Josh Alman and
                   Virginia Vassilevska Williams},
  title        = {{Truly Subquadratic 3SUM and Truly Subcubic APSP via Triangles in Sparse Lopsided Graphs}},
  journal      = {CoRR},
  volume       = {abs/2610.06783},
  year         = {2026},
  doi          = {10.48550/ARXIV.2610.06783},
}

@inproceedings{GK25,
  author       = {Geri Gokaj and
                   Marvin K{\"{u}}nnemann},
  title        = {{Completeness Theorems for $k$-SUM and Geometric Friends: Deciding Fragments of Linear Integer Arithmetic}},
  booktitle    = {{16th Innovations in Theoretical Computer Science Conference (ITCS 2025)}},
  series       = {LIPIcs},
  volume       = {325},
  publisher    = {Schloss Dagstuhl - Leibniz-Zentrum f{\"{u}}r Informatik},
  pages        = {55:1--55:25},
  year         = {2025},
  doi          = {10.4230/LIPIcs.ITCS.2025.55},
  url          = {https://doi.org/10.4230/LIPIcs.ITCS.2025.55},
  bibsource    = {dblp computer science bibliography, https://dblp.org}
}

@article{BFK+26,
  author       = {Daniel Bertschinger and
                   Henry F{\"{o}}rster and
                   Fabian Klute and
                   Irene Parada and
                   Patrick Schnider and
                   Birgit Vogtenhuber},
  title        = {{Bowties and Hourglasses: Intersections of Double-Wedges or: Stabbing and Avoiding Line Segments}},
  journal      = {Inf. Process. Lett.},
  volume       = {194},
  pages        = {106642},
  year         = {2026},
  doi          = {10.1016/j.ipl.2026.106642},
  url          = {https://doi.org/10.1016/j.ipl.2026.106642},
  bibsource    = {dblp computer science bibliography, https://dblp.org}
}

@inproceedings{LCA26,
  author       = {Jaegun Lee and
                   Chaeyoon Chung and
                   Hee-Kap Ahn},
  title        = {{Bichromatic Classifications of Points Using Strips}},
  booktitle    = {{20th Scandinavian Symposium on Algorithm Theory (SWAT 2026)}},
  series       = {LIPIcs},
  volume       = {370},
  publisher    = {Schloss Dagstuhl - Leibniz-Zentrum f{\"{u}}r Informatik},
  pages        = {29:1--29:17},
  year         = {2026},
  doi          = {10.4230/LIPIcs.SWAT.2026.29},
  url          = {https://doi.org/10.4230/LIPIcs.SWAT.2026.29},
  bibsource    = {dblp computer science bibliography, https://dblp.org}
}

@misc{OpenAI26,
  author       = {{OpenAI}},
  title        = {{Quasipolynomial Bounds for Arithmetic Progressions}},
  year         = {2026},
  howpublished = {Preprint available at \href{https://github.com/openai/math/blob/main/preprints/Quasipolynomial-Bounds-for-Arithmetic-Progressions-September-23-2026/paper.pdf}{github.com/openai/math}}
}

@inproceedings{Veg90,
  author       = {Gert Vegter},
  title        = {{The Visibility Diagram: a Data Structure for Visibility Problems and Motion Planning}},
  booktitle    = {{2nd Scandinavian Workshop on Algorithm Theory (SWAT 1990)}},
  series       = {Lecture Notes in Computer Science},
  volume       = {447},
  publisher    = {Springer},
  pages        = {97--110},
  year         = {1990},
  doi          = {10.1007/3-540-52846-6_81},
  url          = {https://doi.org/10.1007/3-540-52846-6_81},
  bibsource    = {dblp computer science bibliography, https://dblp.org}
}

@article{CK93,
  author       = {L. Paul Chew and
                   Klara Kedem},
  title        = {{A Convex Polygon Among Polygonal Obstacles: Placement and High-clearance Motion}},
  journal      = {Comput. Geom.},
  volume       = {3},
  pages        = {59--89},
  year         = {1993},
  doi          = {10.1016/0925-7721(93)90001-M},
  url          = {https://doi.org/10.1016/0925-7721(93)90001-M},
  bibsource    = {dblp computer science bibliography, https://dblp.org}
}

@article{HS96,
  author       = {Dan Halperin and
                   Micha Sharir},
  title        = {{A Near-Quadratic Algorithm for Planning the Motion of a Polygon in a Polygonal Environment}},
  journal      = {Discret. Comput. Geom.},
  volume       = {16},
  number       = {2},
  pages        = {121--134},
  year         = {1996},
  doi          = {10.1007/BF02716803},
  url          = {https://doi.org/10.1007/BF02716803},
  bibsource    = {dblp computer science bibliography, https://dblp.org}
}

@inproceedings{GKST25,
  author       = {Geri Gokaj and
                   Marvin K{\"{u}}nnemann and
                   Sabine Storandt and
                   Carina Truschel},
  title        = {{(Multivariate) $k$-SUM as Barrier to Succinct Computation}},
  booktitle    = {{33rd European Symposium on Algorithms (ESA 2025)}},
  series       = {LIPIcs},
  volume       = {351},
  publisher    = {Schloss Dagstuhl - Leibniz-Zentrum f{\"{u}}r Informatik},
  pages        = {42:1--42:19},
  year         = {2025},
  doi          = {10.4230/LIPIcs.ESA.2025.42},
  url          = {https://doi.org/10.4230/LIPIcs.ESA.2025.42},
  bibsource    = {dblp computer science bibliography, https://dblp.org}
}

@inproceedings{ABB+26,
  author       = {Mikkel Abrahamsen and
                   Sujoy Bhore and
                   Maike Buchin and
                   Jacobus Conradi and
                   Ce Jin and
                   Andr{\'{e}} Nusser and
                   Carolin Rehs},
  title        = {{Algorithms and Lower Bounds for the Maximum Overlap of Two Polygons Under Translation}},
  booktitle    = {{37th ACM-SIAM Symposium on Discrete Algorithms (SODA 2026)}},
  publisher    = {SIAM},
  pages        = {5624--5640},
  year         = {2026},
  doi          = {10.1137/1.9781611978971.200},
  url          = {https://doi.org/10.1137/1.9781611978971.200},
  bibsource    = {dblp computer science bibliography, https://dblp.org}
}

@inproceedings{KN22,
  author       = {Marvin K{\"{u}}nnemann and
                   Andr{\'{e}} Nusser},
  title        = {{Polygon Placement Revisited: (Degree of Freedom + 1)-SUM Hardness and an Improvement via Offline Dynamic Rectangle Union}},
  booktitle    = {{33rd ACM-SIAM Symposium on Discrete Algorithms (SODA 2022)}},
  publisher    = {SIAM},
  pages        = {3181--3201},
  year         = {2022},
  doi          = {10.1137/1.9781611977073.124},
  url          = {https://doi.org/10.1137/1.9781611977073.124},
  bibsource    = {dblp computer science bibliography, https://dblp.org}
}

\end{document}